\documentclass[journal,twocolumns]{IEEEtran}

\IEEEoverridecommandlockouts
\usepackage{amsfonts}
\usepackage{amsmath}
\usepackage{amsthm}
\usepackage{amssymb}
\usepackage{graphicx}
\usepackage[T1]{fontenc}
\usepackage{supertabular}
\usepackage{longtable}
\usepackage[usenames,dvipsnames]{color}
\usepackage{bbm}
\usepackage{fancyhdr}
\usepackage{breqn}
\usepackage{fixltx2e}
\usepackage{capt-of}
\usepackage[framemethod=TikZ]{mdframed}
\usepackage{xcolor}
\usepackage{tikz}
\usepackage{endnotes}
\usepackage{supertabular}
\usepackage{tabularx}
\usepackage{lipsum}
\usepackage{soul}

\newtheorem{remark}{Remark}

\newtheorem{proposition}{Proposition}
\newtheorem{definition}{Definition}

\newcommand{\mathsym}[1]{{}}
\newcommand{\unicode}[1]{{}}

\usepackage{grffile}
\mdfsetup{%
   middlelinecolor=black,
   middlelinewidth=1pt,
   backgroundcolor=yellow!40,
   roundcorner=5pt}

\begin{document}
\title{\color{Brown} Stochastic Tail Exponent For Asymmetric Power Laws\\
}
\author{Nassim Nicholas Taleb\\
Tandon School of Engineering, New York University\\ 
}

\maketitle
\thispagestyle{fancy}
\markboth{\textbf{Fat Tails Statistics Project}}
\flushbottom 
\begin{abstract}
We examine random variables in the power law/regularly varying class with stochastic tail exponent, the exponent $\alpha$ having its own distribution. We show the effect of stochasticity of $\alpha$ on the expectation and higher moments of the random variable. For instance, the moments of a right-tailed  or right-asymmetric variable, when finite, increase with the variance of $\alpha$; those of a left-asymmetric one decreases. The same applies to conditional shortfall (CVar), or mean-excess functions.

We prove the general case and examine the specific situation of lognormally distributed $\alpha \in [b,\infty), b>1$. 

The stochasticity of the exponent induces a significant bias in the estimation of the mean and higher moments in the presence of data uncertainty.  This has consequences on sampling error as uncertainty about $\alpha$ translates into a higher expected mean.

The bias is conserved under summation, even upon large enough a number of summands to warrant convergence to the stable distribution. We establish inequalities related to the asymmetry.

We also consider the situation of capped power laws (i.e. with compact support), and apply it to the study of violence by Cirillo and Taleb (2016). We show that uncertainty concerning the historical data increases the true mean.
\end{abstract}
\thanks{Conference: Extremes and Risks in Higher Dimensions, Lorentz Center,
Leiden, The Netherlands, September 2016.}
\section{Background}
Stochastic volatility has been introduced heuristically in mathematical finance by traders looking for biases on option valuation, where a Gaussian distribution is considered to have several possible variances, either locally or at some specific future date. Options far from the money (i.e. concerning tail events) increase in value with uncertainty on the variance of the distribution, as they are convex to the standard deviation.

This led to a family of models of Brownian motion with stochastic variance (see review in Gatheral \cite{gatheral2006volatility}) and proved useful in tracking the distributions of the underlying and the effect of the nonGaussian character of random processes on functions of the process (such as option prices).

Just as options are convex to the scale of the distribution, we find many situations where expectations are convex to the power law tail exponent.
This note examines two cases:
\begin{itemize}
\item The standard power laws, one-tailed or asymmetric (with support that includes either $+\infty$ or $-\infty$.
\item The pseudo-power law, where a random variable appears to be a power law but has compact support, as in the study of violence \cite{cirillo2016statistical} where wars have the number of casualties capped at a maximum value.
\end{itemize}

\section{One Tailed Distributions with Stochastic Alpha}
\subsection{General Cases}
\begin{definition}
Let $X$ be a random variable belonging to the class of distributions with a
"power law" right tail, that is support in $[x_0,+\infty)\,,x_0 \in \mathbb{R}$: 

Subclass $\mathfrak{P}_1$:
\begin{equation}
\{X: \mathbb{P}(X>x)= L(x) x^{-\alpha } ,L'(x)=0 \} \label{powerlaweq}
\end{equation}%


Class $\mathfrak{P}$:
\begin{equation}
\{X:\mathbb{P}(X>x)\sim L(x)\,x^{-\alpha } \} \label{powerlaweq2}
\end{equation}%
where $\sim$ means that the limit of the ratio or rhs to lhs goes to 1 as $x \to \infty$. $L:\left[ x_{\min },+\infty \right) \rightarrow \left( 0,+\infty
\right) $ is a slowly varying function, defined as $\lim_{x\rightarrow
+\infty }\frac{L(kx)}{L(x)}=1$ for any $k>0$. The constant $\alpha >0$.

We further assume that:

\begin{align} 
\lim_{x\to \infty} L'(x) \, x &=0 \label{firstderiv}\\
\lim_{x\to \infty} L''(x) \,x^2 &=0\label{secondderiv}
\end{align}


We have $$\mathfrak{P}_1 \subset \mathfrak{P}$$
\end{definition}

We note that the first class corresponds to the Pareto distributions (with proper shifting and scaling), where $L$ is a constant and $\mathfrak{P} \setminus \mathfrak{P_2}$ to the more general one-sided, Beta Prime, or half-Student T .  As to $\mathfrak{P_2}\setminus\mathfrak{P_1}$ we can include all manner of mixed distributions.

\subsection{Stochastic Alpha Inequality}

Throughout the rest of the paper we use for notation $X'$ for the stochastic alpha version of $X$, the constant $\alpha$ case.

\begin{proposition}
Let $p= 1,2,...$, $X'$ be the same random variable as $X$ above in $\mathfrak{P}_1$ (the one-tailed regular variation class), with $x_0\geq 0$, except with stochastic $\alpha$ with all realizations $>p$ that preserve the mean $\bar{\alpha}$,
	$$\mathbb{E}(X^{'p}) \geq \mathbb{E}(X^p).$$\label{propositionp1}
\end{proposition}
The results extend 
to the general class $\mathfrak{P}$ under some approximations of the mean.
\begin{proposition}
Let $K$ be a threshold. With $X$ in the $\mathfrak{P}$ class, we have	the expected conditional shortfall (CVar):
$$\lim_{K\to \infty} \mathbb{E}(X^'|_{X'>K}) \geq \lim_{K\to \infty}\mathbb{E}(X|_{X>K}).$$
\end{proposition}

\begin{proof}

We remark that $\mathbb{E}(X^p)$ is convex to $\alpha$, in the following sense:  let $\alpha_i>p \, \forall i$, the weights $\omega_i$:
$\sum_i\omega_i =1$, $0\leq|\omega_i|\leq 1$, 
$\sum_i \omega_i \alpha_i=\bar{\alpha}$, Jensen's inequality is expressed as:
$$\sum_i \omega_i\mathbb{E}(X_{\alpha_i}) \geq \mathbb{E}(\sum_i(\omega_i X_{\alpha_i})).$$

We first need to solve for the density: $\varphi(x)=\alpha  x^{-\alpha -1} L(x,\alpha )-x^{-\alpha } L^{(1,0)}(x,\alpha )$ and get the normalizing constant.

\begin{equation}
L(x_0 ,\alpha )=x_0 ^{\alpha }-\frac{2 x_0  L^{(1,0)}(x_0 ,\alpha )}{\alpha -1}-\frac{2 x_0 ^2 L^{(2,0)}(x_0 ,\alpha )}{(\alpha -1) (\alpha -2)},\label{normalizing}	
\end{equation}
$\alpha \not =1,2$, where the slot notation $L^{(1,0)}(x_0 ,\alpha )$ is short for $\frac{\partial L(x,\alpha)}{\partial x}|_{x=x_0}$.

By the Karamata representation theorem, \cite{karamata1932inegalite},\cite{bingham1989regular},\cite{teugels1975class}, a function $L$ on $[x_0,+\infty)$ is slowly moving if and only if it can be written in the form $L(x)=\exp \left(\int_{x_0}^x \frac{\epsilon(t)}{t} \, dt \right)+\eta (x)$ where 
$\eta(.)$ is a bounded measurable function converging to a finite number as $x \rightarrow +\infty$, and $\epsilon(x)$ is a bounded measurable function converging to zero as $x\rightarrow +\infty$. 

Accordingly, $L'(x)$ goes to $0$ as $x \to \infty$. (We further assumed in \ref{firstderiv} and \ref{secondderiv} that $L'(x)$ goes to 0 faster than $x$ and $L^{''}(x)$ goes to 0 faster than $x^2$.) Integrating by parts, 
$$\mathbb{E}(X^{p})=x_0^p+ p \int_{x_0}^{\infty }  x^{p-1} \, \mathrm{d}\bar{F}(x)$$
where $\bar{F}$ is the survival function in Eqs. \ref{powerlaweq} and \ref{powerlaweq2}. Integrating by parts $3$ additional times and eliminating derivatives of higher order than $2$:
\begin{dmath}
\mathbb{E}(X^{p})=\frac{x_0 ^{p-\alpha } L(x_0 ,\alpha )}{p-\alpha }-\frac{x_0 ^{p-\alpha+1} L^{(1,0)}(x_0 ,\alpha )}{(p-\alpha ) (p-\alpha+1)}+\frac{x_0 ^{p-\alpha+2} L^{(2,0)}(x_0 ,\alpha )}{(p-\alpha ) (p-\alpha+1) (p-\alpha+2)}	\label{condition}
\end{dmath}
which, for the special case of $X$ in $\mathfrak{P}_1$ reduces to:
\begin{dmath}
\mathbb{E}(X^{p})=x_0 ^{p} \frac{\alpha}{\alpha -p}\label{expectation}
\end{dmath}

As to Proposition 2, we can prove it simply from the properties that $\lim_{x \to \infty} L'(x)=0$. This allows a proof of var der Mijk's law that Paretian inequality is invariant to the threshold in the tail, that is $\frac{\mathbb{E}(X|_{X>K})}{K}$ converges to a constant.
\end{proof}
Equation \ref{condition} presents the exact conditions on the functional form of $L(x)$ for the convexity to extend to sub-classes of $\mathfrak{P}$.

 %
%
%
%
Our results hold to distributions that are transformed by shifting and scaling, of the sort:
 
 $x \mapsto x-\mu +x_0$ (Pareto II), or with further transformations to Pareto types II and IV.

We note that the representation $\mathfrak{P}_1$ uses the same parameter, $x_0$, for both scale and minimum value, as a simplification.  
 
We can verify that the expectation from Eq. \ref{expectation} is convex to $\alpha$:
$\frac{\partial \mathbb{E}(X^{p})}{\partial \alpha ^2}=x_0 ^{p} \frac{2}{(\alpha -1)^3}$.\subsection{Approximations for the Class $\mathfrak{P}$}
For $\mathfrak{P} \setminus \mathfrak{P_1}$, our results hold when we can write an approximation of the expectation of $X$ as a constant multiplying the integral of $x^{-\alpha}$, namely 
\begin{equation}
	\mathbb{E}(X)\approx k \frac{\nu(\alpha) }{\alpha -1}\label{approxpl}
\end{equation}
where $k$ is a positive constant that does not depend on $\alpha$ and $\nu(.)$ is approximated by a linear function of $\alpha$ (plus a threshold). The expectation will be convex to $\alpha$.


\subsubsection*{Example: Student T Distribution}
 For the Student T distribution 
with tail $\alpha$, the  "sophisticated" slowly varying function in common use for symmetric power laws in quantitative finance, the half-mean or the mean of the one-sided distribution (i.e. with support on $\mathbb{R}^+$ becomes $$2  \nu(\alpha)=2 \frac{\sqrt{\alpha } \Gamma \left(\frac{\alpha +1}{2}\right)}{\sqrt{\pi } \Gamma \left(\frac{\alpha }{2}\right)}\approx \alpha \frac{  (1+\log (4))}{\pi },$$
where $\Gamma(.)$ is the gamma function.
\section{Sums of Power Laws}
As we are dealing from here on with convergence to the stable distribution, we consider situations of $1<\alpha<2$, hence $p=1$ and will be concerned solely with the mean.

We observe that the convexity of the mean is invariant to summations of power law distributed variables as $X$ above. The Stable distribution has a mean that in conventional parameterizations does not appear to depend on $\alpha$ --but in fact depends on it.

Let $Y$ be distributed according to a Pareto distribution with density $f(y)\triangleq \alpha  \lambda^{\alpha } y^{-\alpha -1}, y\geq \lambda > 0$ and with its tail exponent $1<\alpha<2$. 
Now, let $Y_1,Y_2,\ldots Y_n$ be identical and independent copies of $Y$. 
Let $\chi(t)$ be the characteristic function for $f(y)$. 
We have $\chi(t)=\alpha  (-i t)^{\alpha } \Gamma (-\alpha ,-i t)$, where $\gamma(.,.)$ is the incomplete gamma function. We can get the mean from the characteristic function of the average of $n$ summands $\frac{1}{n} (Y_1+Y_2+...Y_n)$, namely $ \chi (\frac{t}{n})^n$. Taking the first derivative:
\begin{dmath}
-i \frac{\partial \chi (\frac{t}{n})^n}{\partial t}=(-i)^{\alpha  (n-1)} n^{1-\alpha  n} \alpha ^n \lambda ^{\alpha  (n-1)} t^{\alpha  (n-1)-1} \Gamma \left(-\alpha ,-\frac{i t \lambda }{n}\right)^{n-1} \left((-i)^{\alpha } \alpha  \lambda ^{\alpha } t^{\alpha } \Gamma \left(-\alpha ,-\frac{i t \lambda }{n}\right)-n^{\alpha } e^{\frac{i \lambda  t}{n}}\right)
\end{dmath}
and 
\begin{equation}
\lim_{n\to \infty}	-i \frac{\partial \chi (\frac{t}{n})^n}{\partial t}\bigg\arrowvert_{t=0}=\lambda\frac{\alpha}{\alpha-1}
\end{equation}
Thus we can see how the converging asymptotic distribution for the average will have for mean the scale times $\frac{\alpha}{\alpha-1}$, which does not depends on $n$.

Let $\chi^S(t)$ be the characteristic function of the corresponding stable distribution $S_{\alpha, \beta,\mu, \sigma}$,  from the distribution of an infinitely summed copies of $Y$. 
By the L\'evy continuity theorem, we have 
\begin{itemize}
\item $\frac{1}{n}  \Sigma_{i\leq n} Y_i  \xrightarrow{\mathcal{D}}\ S$, with distribution $S_{\alpha, \beta,\mu, \sigma}$, where $\xrightarrow{\mathcal{D}}$ denotes convergence in distribution 

and

\item $ \chi^S(t)=\lim_{n \to\infty} \chi(t/n)^n$

\end{itemize}
are equivalent.

So we are dealing with the standard result \cite{zolotarev1971new},\cite{samorodnitsky1994stable}, for exact Pareto sums \cite{zaliapin2005approximating}, replacing the conventional $\mu$ with the mean from above:
$$\chi^S(t)=\exp \left(i \left(\lambda \frac{\alpha  t}{\alpha -1}+\left| t\right| ^{\alpha } \left(\beta  \tan \left(\frac{\pi  \alpha }{2}\right) \text{sgn}(t)+i\right)\right)\right).$$

\section{Asymmetric Stable Distributions}
We can verify by symmetry that, effectively, flipping the distribution 
around $y_0$ and replacing $+\infty$ with $-\infty$ yields a negative value of the mean and higher (existing) moments, hence a degradation effect from stochastic $\alpha$.

The central question becomes: 
\begin{remark} [Preservation of Asymmetry]
A normalized sum in $\mathfrak {P_1}$ one-tailed distribution with expectation that depends on $\alpha$ of the form in Eq. \ref{approxpl} will necessarily converge in distribution to an asymmetric stable distribution $S_{\alpha,\beta,\mu,1}$, with $ \beta \not = 0$.	
\end{remark}

\begin{remark}
Let $Y'$ be $Y$ under mean-preserving stochastic $\alpha$. The convexity effect, or $\text{sgn}\left(\mathbb{E}(Y')-\mathbb{E}(Y)\right) = \text{sgn}(\beta)$.	
\end{remark}
\begin{proof}
Consider two slowly moving functions as in \ref{powerlaweq}, each on one side of the tails. We have $L(y)= \mathbbm{1}_{y< y_\theta}  L^-(y)+\mathbbm{1}_{y \geq y_\theta} L^+(y)$:

$$\begin{cases}
	L^+(y), L: [y_\theta,+\infty],& \lim_{y\rightarrow \infty}L^+(y)=c\\
	\\
	L^-(y),L: [-\infty,y_{\theta}],& \lim_{y\rightarrow -\infty}L^-(y)=d.
\end{cases}
$$
From \cite{samorodnitsky1994stable},

\bigskip
if $\begin{cases}
 \mathbb{P}(X>x) \sim c x^{-\alpha}, x \rightarrow +\infty \\
 	\\
  \mathbb{P}(X<x) \sim d |x|^{-\alpha}, x \rightarrow +\infty 	,
 	
 \end{cases}$
 then $Y$ converges in distribution to $S_{\alpha,\beta,\mu,1}$ with the coefficient $\beta=\frac{c-d}{c+d}$. 
 
We can show that the mean can be written as $(\lambda_+-\lambda_-) \frac{\alpha}{\alpha-1}$ where:
 $$\lambda_+ \geq \lambda_- \text{ if } \int_{y_\theta}^{\infty} L^+(y) \mathrm{d}y, \, \geq\int_{-\infty}^{y_\theta} L^-(y) \mathrm{d}y$$ 
\end{proof}

\section{Pareto Distribution with lognormally distributed $\alpha$}

Now assume $\alpha$ is  following a shifted Lognormal distribution with mean $\alpha_0$ and minimum value $b$, that is, $\alpha-b$ follows a Lognormal $\mathcal{LN}\left(\log (\alpha_0)-\frac{\sigma ^2}{2},\sigma \right)$.
The parameter $b$ allows us to work with a lower bound on the tail exponent in order to satisfy finite expectation. We  know that the tail exponent will eventually converge to $b$ but the process may be quite slow. 

\begin{proposition}
Assuming finite expectation for X' and for exponent the lognormally distributed shifted variable $\alpha-b$ with law $\mathcal{LN}\left(\log (\alpha_0)-\frac{\sigma ^2}{2},\sigma \right)$,  $b \geq 1$ mininum value for $\alpha$,  and scale $\lambda$:\begin{dmath}
		\mathbb{E}(Y')= 	\mathbb{E}(Y) +\lambda \frac{( e^{\sigma ^2}-b)}{\alpha_0-b}
	\end{dmath}

\end{proposition}

We need $b \geq 1$ to avoid problems of infinite expectation. 

Let $\phi(y,\alpha)$ be the density with stochastic tail exponent. With $\alpha>0,  \alpha_0>b, b \geq 1, \sigma>0, Y \geq \lambda >0$ ,

%
%
\begin{dmath}
\mathbb{E}(Y)=	\int_b^\infty  \int_L^\infty y \phi(y;\alpha)\,\mathrm{d}y \,\mathrm{d}\alpha
= \int_b^\infty \lambda \frac{\alpha}{\alpha-1} \frac{1}{\sqrt{2 \pi } \sigma  (\alpha -b)}\\
\exp \left(-\frac{\left(\log  (\alpha -b)-\log  (\alpha_0-b)+\frac{\sigma ^2}{2}\right)^2}{2 \sigma ^2}\right) \,\mathrm{d}\alpha\\
= \frac{\lambda \left(\alpha_0+e^{\sigma ^2}-b\right)}{\alpha_0-b}
\end{dmath}.

\subsection*{Approximation of the density}
With $b=1$ (which is the lower bound for $b$),we get the density with stochastic $\alpha$:
\begin{dmath}
	\phi(y;\alpha_0,\sigma)=\lim_{k \to \infty} \frac{1}{Y^2 }\sum _{i=0}^k \frac{1}{i!}L (\alpha_0-1)^i e^{\frac{1}{2} i (i-1) \sigma ^2} (\log (\lambda)-\log (y))^{i-1} (i+\log (\lambda)-\log (y))
	\end{dmath}
This result is obtained by expanding $\alpha$ around its lower bound $b$ (which we simplified to $b=1$) and integrating each summand.
\section{Pareto Distribution with Gamma distributed Alpha}
\begin{proposition}
Assuming finite expectation for $X'$ scale $\lambda$, and for exponent a gamma distributed shifted variable $\alpha-1$ with law $\varphi(.)$, mean $\alpha_0$ and variance $s^2$,  all values for $\alpha$ greater than 1:
 	\begin{equation}
\mathbbm{E}(X')=\mathbbm{E}(X')+\frac{s^2}{(\alpha_0-1) (\alpha_0-s-1) (\alpha_0+s-1)}
\end{equation}
\end{proposition}

\begin{proof}
	
\begin{equation}
\begin{array}{cc}
 \varphi(\alpha)=\frac{e^{-\frac{(\alpha -1) \left(\alpha _0-1\right)}{s^2}}
   \left(\frac{s^2}{(\alpha -1) \left(\alpha
   _0-1\right)}\right){}^{-\frac{\left(\alpha
   _0-1\right){}^2}{s^2}}}{(\alpha -1) \Gamma
   \left(\frac{\left(\alpha _0-1\right){}^2}{s^2}\right)}, & \alpha >1\\
  \end{array}
\end{equation}

\begin{equation}
\int_1^{\infty } \alpha  \lambda ^{\alpha } x^{-\alpha -1} \varphi(\alpha) \, d\alpha	
\end{equation}

$$=\int_1^{\infty } \frac{\alpha  \left(e^{-\frac{(\alpha -1) (\alpha_0-1)}{s^2}} \left(\frac{s^2}{(\alpha -1) (\alpha_0-1)}\right)^{-\frac{(\alpha_0-1)^2}{s^2}}\right)}{(\alpha -1) \left((\alpha -1) \Gamma \left(\frac{(\alpha_0-1)^2}{s^2}\right)\right)} \, d\alpha$$
$$=\frac{1}{2} \left(\frac{1}{\alpha_0+s-1}+\frac{1}{\alpha_0-s-1}+2\right)$$
\end{proof}

\section{The bounded power law in Cirillo and Taleb (2016)}
In \cite{cirillo2016statistical} and \cite{cirillo2016expected}, the studies make use of bounded power laws, applied to violence and operational risk, respectively. Although with $\alpha<1$ the variable $Z$ has finite expectations owing to the upper bound. 

The methods offered were a smooth transformation of the variable as follows: we start with $z \in [L,H), L>0$ and transform it into  $x \in [L,\infty)$, the latter legitimately being power law distributed.

So the smooth logarithmic transformation):
$$x=\varphi(z)=L-H \log \left(\frac{H-x}{H-L}\right),$$
and
$$f(x)=\frac{\left(\frac{x-L}{\alpha  \sigma }+1\right)^{-\alpha -1}}{\sigma }.$$
We thus get the distribution of $Z$ which will have a finite expectation for all positive values of $\alpha$.

\begin{dmath}
\frac{\partial^2 E(Z)}{\partial \alpha^2}=\frac{1}{H^3}	(H-L) \left(e^{\frac{\alpha  \sigma }{H}} \left(2 H^3
   G_{3,4}^{4,0}\left(\frac{\alpha  \sigma }{H}|
\begin{array}{c}
 \alpha +1,\alpha +1,\alpha +1 \\
 1,\alpha ,\alpha ,\alpha  \\
\end{array}
\right)\\-2 H^2 (H+\sigma ) G_{2,3}^{3,0}\left(\frac{\alpha  \sigma
   }{H}|
\begin{array}{c}
 \alpha +1,\alpha +1 \\
 1,\alpha ,\alpha  \\
\end{array}
\right)\\+\sigma  \left(\alpha  \sigma ^2+(\alpha +1) H^2+2 \alpha 
   H \sigma \right) E_{\alpha }\left(\frac{\alpha  \sigma
   }{H}\right)\right)-H \sigma  (H+\sigma )\right)
\end{dmath}
which appears to be positive in the range of numerical perturbations in \cite{cirillo2016statistical}.\footnote{
$
 G_{3,4}^{4,0}\left(\frac{\alpha  \sigma }{H}|
\begin{array}{c}
 \alpha +1,\alpha +1,\alpha +1 \\
 1,\alpha ,\alpha ,\alpha  \\
\end{array}
\right)$
%
is the Meijer G function.}
 At such a low level of $\alpha$, around $\frac{1}{2}$, the expectation is extremely convex and the bias will be accordingly extremely pronounced.

This convexity has the following practical implication. Historical data on violence over the past two millennia, is fundamentally unreliable \cite{cirillo2016statistical}. Hence an imprecision about the tail exponent, from errors embedded in the data, need to be present in the computations. The above shows that uncertainty about  $\alpha$, is more likely to make the "true" statistical mean (that is the mean of the process as opposed to sample mean) higher than lower, hence supports the statement that more uncertainty increases the estimation of violence.
\section{Additional Comments}

The bias in the estimation of the mean and shortfalls from uncertainty in the tail exponent can be added to analyses where data is insufficient, unreliable, or simply prone to forgeries.

In additional to statistical inference, these result can extend to processes, whether a compound Poisson process with power laws subordination \cite{stam1973regular} (i.e. a Poisson arrival time and a jump that is power law distributed) or a L\'evy process. The latter can be analyzed by considering successive "slice distributions" or discretization of the process \cite{tankov2003financial}. Since the expectation of a sum of jumps is the sum of expectation, the same convexity will appear as the one we got from Eq. \ref{approxpl}.
\section{Acknowledgments}
Marco Avellaneda, Robert Frey, Raphael Douady, Pasquale Cirillo.

\bibliographystyle{IEEEtran}
\bibliography{/Users/nntaleb/Dropbox/Central-bibliography}

\end{document}